\documentclass[11pt,a4paper]{article}

\PassOptionsToPackage{colorlinks=true, linkcolor=black, citecolor=blue!50!black, urlcolor=blue!50!black}{hyperref}

\usepackage{float}        
\usepackage{placeins}    
\usepackage{booktabs}    
\usepackage{xcolor}

\usepackage{geometry}
\usepackage[utf8]{inputenc}
\usepackage[T1]{fontenc}
\usepackage[english]{babel}
\usepackage{amsmath,amssymb,amsthm}
\usepackage{bm}
\usepackage{enumitem}

\usepackage{graphicx}
\usepackage{subcaption}  

\usepackage[round]{natbib}

\usepackage{authblk}

\usepackage{hyperref}

\newtheorem{proposition}{Proposition}
\newtheorem{remark}{Remark}

\newcommand{\bx}{\bm{x}}
\newcommand{\T}{^{\mathsf{T}}}
\newcommand{\Yhat}{\widehat{Y}}
\newcommand{\Nhat}{\widehat{N}}
\newcommand{\Xhat}{\widehat{\bm{X}}}
\newcommand{\E}{\mathbb{E}}
\newcommand{\Var}{\mathrm{Var}}
\newcommand{\Vhat}{\widehat{\Var}}
\newcommand{\mh}{\widehat{m}}
\newcommand{\Bhat}{\widehat{\bm{B}}}

\title{{\bfseries Direct domain estimation via regression-tree-assisted\\ estimators in the production of official statistics}}

\author{Juan Pablo Ferreira \\ \vspace{5pt}
\small Facultad de Ciencias Económicas y de Administración, Universidad de la República, Uruguay\\
\small \texttt{juanpablo.ferreira@fcea.edu.uy}}
\date{January 15, 2025}

\AtBeginDocument{
  \setlength{\abovedisplayskip}{10pt}       
  \setlength{\belowdisplayskip}{10pt}
  \setlength{\abovedisplayshortskip}{8pt}
  \setlength{\belowdisplayshortskip}{8pt}
}

\begin{document}

\maketitle
\vspace{-2em} 

\begin{abstract}
\noindent
National statistical offices (NSOs) produce their estimates under a single weighting system (the \emph{uni-weight} approach): one set of weights, independent of the variable of interest, is used to estimate multiple parameters and multiple subpopulations (domains). In this paper we study, within the family of model-assisted estimators and from a design-based perspective of \emph{direct} estimation, the use of regression trees as the assisting model for estimating totals in unplanned domains. We distinguish two strategies: (i) fitting a single tree at the population level and deriving from it uni-weight weights applicable to any domain (RT$_U$), and (ii) fitting a domain-specific tree (RT$_D$). We show that both estimators can be written as weighted sums with weights that do not depend on $y$, preserving the uni-weight property and additivity (\emph{benchmarking}) with respect to the population total. Extending to trees the classical result of \citet{estevao1999}, we argue why the estimator built from a population-level model tends to behave like the Horvitz--Thompson estimator within domains, whereas the domain-specific model can achieve substantial variance reductions. A simulation study based on microdata from the Uruguayan Continuous Household Survey (ECH) illustrates the behavior of the estimators at the population level and by department.\\[4pt]
\end{abstract}

\section{Introduction}

Sample surveys are one of the main sources of information for the production of official statistics. In the practice of national statistical offices (NSOs), estimates are obtained by attaching a weight to each sampled unit and using a \emph{single} weighting system to estimate different parameters (totals, means, ratios) and different subsets of the population. This approach, known as \emph{uni-weight} \citep{estevao1999}, has decisive operational advantages: it simplifies the computation of point estimates, creates economies of scale, suits continuous surveys that must release results in a timely fashion, and guarantees internal consistency among the figures disseminated. The weights are typically built through regression or calibration estimators \citep{deville1992,sarndal2007}, accounting for production needs.

The uni-weight approach is not, however, optimal for every parameter or domain. The efficiency of a model-assisted estimator depends on the predictive power of the assisting model for the variable of interest \citep{sarndal1992,breidt2017}; a single model, fitted to the whole population, does not necessarily capture the relationship between the auxiliary variables and the study variable within a particular subdomain. This tension between operational convenience and local efficiency is the starting point of this work.

\paragraph{Modern assisting models.} The recent literature has incorporated flexible prediction techniques as assisting models within the design-based framework. \citet{breidt2017} review the use of modern prediction methods. \citet{mcconville2019} propose an estimator assisted by a regression tree and, building on \citet{toth2011}, establish its design consistency; the estimator admits an appealing interpretation as a post-stratified estimator whose post-strata are selected automatically by the recursive partitioning algorithm. \citet{dagdoug2023} extend the approach to random forests and, relevantly for this work, distinguish estimators based on partitions built at the population level from those based on local partitions. The use of the lasso as a variable-selection mechanism was studied by \citet{mcconville2017}.

\paragraph{Domain estimation.} The problem of estimating parameters for subpopulations admits two broad families of solutions. On one hand, \emph{model-based} or indirect small area estimation (SAE), based on random-effects models that borrow strength across domains \citep{rao2015}; recent work incorporates machine learning into this framework through mixed models of the random-forest type \citep{frink2024}. On the other hand, \emph{design-based} or direct estimation, which uses only the observations within the domain (possibly augmented with auxiliary information) and preserves design properties such as asymptotic unbiasedness and design-based variance estimation. \citet{sarndal1984} compares both strategies as functions of sample size, sampling fraction, and departure from the model; \citet{estevao2004} argue that, for a wide class of design-consistent estimators and under general conditions, ``borrowing strength'' is not the best technique. The comparison between area- and unit-level estimators has been systematized by \citet{hidiroglou2016}.

\paragraph{Contribution of this work.} Our goal is to study, within the uni-weight and design-based direct approach, the use of regression trees as the assisting model for estimating totals in unplanned domains. The contributions are:
\begin{enumerate}[label=(\roman*),leftmargin=2.2em]
\item to formulate the estimator assisted by a domain-specific regression tree (\text{RT}$_D$) and to show that it can be expressed as a linear combination of sample observations with weights independent of the study variable that estimate the domain population size without error;
\item to characterize the estimator built from a population-level tree with uni-weight weights (\text{RT}$_U$), and to explain, extending to trees the argument of \citet{estevao1999}, why it tends to behave like the Horvitz--Thompson estimator within domains;
\item to discuss the \emph{trade-off} between domain-level precision and population-level optimality, and its relevance for large-scale production;
\item to illustrate empirically the behavior of the estimators using microdata from the Uruguayan Continuous Household Survey (ECH).
\end{enumerate}

The remainder of the paper is organized as follows. Section~\ref{sec:marco} fixes the framework and notation. Section~\ref{sec:asistida} reviews model-assisted estimation, the generalized regression estimator, and its reading as calibration. Section~\ref{sec:rt} presents the tree-assisted estimator at the population level. Section~\ref{sec:dominios} develops the two strategies for domain estimation. Section~\ref{sec:sim} reports the simulation study. Section~\ref{sec:disc} discusses limitations and concludes.

\section{Estimation framework and notation}\label{sec:marco}

Let $U=\{1,2,\dots,i,\dots,N\}$ be a finite population of $N$ eligible units (for example, the working-age individuals of a country). The goal is to estimate the total of a variable of interest $y$,
\begin{equation}
Y=\sum_{i\in U} y_i .
\end{equation}
In addition to the study variable, a vector of auxiliary variables $\bx_i$ is available (from censuses, registers, or other sources), known for all units of the frame or, at least, with known population totals $\bm{X}=\sum_{i\in U}\bx_i$.

A sampling design is defined on the basis of the available information (strata, measures of size, etc.) that assigns each unit a positive first-order inclusion probability,
\begin{equation}
\pi_i=\Pr[i\in s] > 0, \qquad \forall\, i\in U,
\end{equation}
and second-order probabilities $\pi_{ij}=\Pr[i,j \in s]$. The sample $s$ is selected and $y_i$ (together with the auxiliary information) is observed for $i\in s$. The basic design weight is $w_i=\pi_i^{-1}$.

\paragraph{Horvitz--Thompson estimator.} The simple estimator uses only the information on the variable of interest from the units included in the sample \citep{horvitz1952}:
\begin{equation}\label{eq:ht}
\Yhat^{\mathrm{HT}}=\sum_{i\in s}\frac{y_i}{\pi_i}=\sum_{i\in s} w_i\, y_i .
\end{equation}
It is design-unbiased, $\E(\Yhat^{\mathrm{HT}})=Y$, with variance
\begin{equation}
\Var(\Yhat^{\mathrm{HT}})=\sum_{i\in U}\sum_{j\in U}(\pi_{ij}-\pi_i\pi_j)\,\frac{y_i}{\pi_i}\,\frac{y_j}{\pi_j},
\end{equation}
and, if the design allows it, an unbiased variance estimator
\begin{equation}
\Vhat(\Yhat^{\mathrm{HT}})=\sum_{i\in s}\sum_{j\in s}\pi_{ij}^{-1}(\pi_{ij}-\pi_i\pi_j)\,\frac{y_i}{\pi_i}\,\frac{y_j}{\pi_j}.
\end{equation}

\section{Model-assisted estimation}\label{sec:asistida}

\subsection{From the difference estimator to the prediction estimator}

Suppose a method $m(\cdot)$ for predicting $y$ is available that \emph{does not depend on the sample}. With its outputs, the difference estimator is defined \citep{cassel1976}:
\begin{equation}\label{eq:diff}
\Yhat^{\mathrm{DIFF}}=\sum_{i\in U} m(\bx_i)+\sum_{i\in s} w_i\bigl(y_i-m(\bx_i)\bigr).
\end{equation}
This estimator is unbiased regardless of the quality of $m(\cdot)$, $\E(\Yhat^{\mathrm{DIFF}})=Y$. Since the first term is not random, its design variance is
\begin{equation}
\Var(\Yhat^{\mathrm{DIFF}})=\sum_{i\in U}\sum_{j\in U}(\pi_{ij}-\pi_i\pi_j)\,\frac{y_i-m(\bx_i)}{\pi_i}\,\frac{y_j-m(\bx_j)}{\pi_j}.
\end{equation}
The variance will be smaller than that of $\Yhat^{\mathrm{HT}}$ to the extent that the residuals $y_i-m(\bx_i)$ have lower variability than the raw data $y_i$. If the method has poor predictive power, the variance approaches that of the Horvitz--Thompson estimator.

The practical problem is that one rarely has an $m(\cdot)$ that is independent of the sample and has good predictive power. Model-assisted estimation solves this by introducing a superpopulation model
\begin{equation}
y_i=m(\bx_i)+\epsilon_i,\qquad \E(\epsilon_i)=0,
\end{equation}
which need not be true for $U$, but whose usefulness increases with its fit. The method $m(\cdot)$ is estimated using the sample data, $\mh(\cdot)$ (which depends on $s$), and substituted into \eqref{eq:diff}, yielding the prediction estimator:
\begin{equation}\label{eq:pred}
\Yhat^{\mathrm{PRED}}=\sum_{i\in U}\mh(\bx_i)+\sum_{i\in s} w_i\bigl(y_i-\mh(\bx_i)\bigr).
\end{equation}
$\Yhat^{\mathrm{PRED}}$ is asymptotically unbiased, $\E(\Yhat^{\mathrm{PRED}})\to Y$. It has no exact variance form; a common estimator is based on the difference-estimator formula, replacing the residuals by $e_i=y_i-\mh(\bx_i)$:
\begin{equation}\label{eq:varpred}
\Vhat(\Yhat^{\mathrm{PRED}})=\sum_{i\in s}\sum_{j\in s}\pi_{ij}^{-1}(\pi_{ij}-\pi_i\pi_j)\,\frac{e_i}{\pi_i}\,\frac{e_j}{\pi_j}.
\end{equation}
The precision of the estimator therefore depends on the predictive power of the model. Among the assisting models that have been used are linear regression \citep{cassel1976}, logistic regression \citep{lehtonen1998}, nonparametric calibration \citep{montanari2005}, regression trees \citep{mcconville2019}, and random forests \citep{dagdoug2023}.

\subsection{The generalized regression estimator and calibration}

When the assisting model is linear, $y_i=\bx_i\T\bm{\beta}+\epsilon_i$ with $\epsilon_i\sim(0,\sigma_i^2)$, the estimator \eqref{eq:pred} takes the form of the generalized regression estimator (GREG):
\begin{equation}
 \Yhat^{\mathrm{GREG}} = \sum_{i\in U}\bx_i\T\Bhat+\sum_{i\in s} w_i\bigl(y_i-\bx_i\T\Bhat\bigr),   
\end{equation}

where
$$
 \Bhat = \Bigl(\sum_{i\in s} w_i\bx_i\bx_i\T/\sigma_i\Bigr)^{-1}\Bigl(\sum_{i\in s} w_i\bx_i y_i/\sigma_i\Bigr).
$$

A key property for production is that it only requires knowing $\bx_i$ for $i\in s$ and the totals $\bm{X}=\sum_{i\in U}\bx_i$ (not the micro-level information for the entire population), so that 

$$\Yhat^{\mathrm{GREG}}=\bm{X}\T\Bhat+\sum_{i\in s} w_i(y_i-\bx_i\T\Bhat)$$.

The estimator is expressed as a weighted sum
\begin{equation}\label{eq:greg-pesos}
\Yhat^{\mathrm{GREG}}=\sum_{i\in s} w_i^*\, y_i,
\end{equation}
where $w_i^*=g_i w_i$ and  

$$g_i=1+(\bm{X}-\Xhat^{\mathrm{HT}})\T\Bigl(\sum_{i\in s} w_i\bx_i\bx_i\T/\sigma_i\Bigr)^{-1}\bx_i/\sigma_i$$

The weights $w_i^*$ \emph{do not depend on $y$} and therefore serve different variables of interest (the uni-weight property). Moreover, the estimator is calibrated to the auxiliary information,
\begin{equation}\label{eq:cal}
\sum_{i\in s} w_i^*\,\bx_i=\sum_{i\in U}\bx_i ,
\end{equation}
which provides coherence with other sources and improves the credibility of the figures for users.

The calibration approach \citep{deville1992,sarndal2007} reaches the same estimator from a different angle: seeking weights $w_i^*=g_i w_i$ that minimize a distance $L(w^*,w)=\sum_{i\in s} L(w_i^*,w_i)$ subject to \eqref{eq:cal}. For the least-squares distance $L(w^*,w)=\sum_{i\in s}(w_i^*-w_i)^2/w_i$ one recovers the GREG, $\Yhat^{\mathrm{CAL}}=\Yhat^{\mathrm{GREG}}$, the standard strategy in NSOs.

\paragraph{Limitations of the GREG.} The estimator may exhibit (i) high variability of the weights $w_i^*$, which unnecessarily inflates standard errors; (ii) negative weights $w_i^*$; and (iii) null precision gains relative to $\Yhat^{\mathrm{HT}}$ if the model is inadequate. Furthermore, $\Yhat^{\mathrm{GREG}}$ can behave poorly when the number of auxiliary variables exceeds what $n$ can support, and post-stratification (a special case of the GREG with categorical $\bx$) becomes fragile when there are many qualitative variables with many categories whose interactions define the cells. These difficulties motivate the use of assisting models that automatically select the relevant structure.

\section{Regression-tree-assisted estimator}\label{sec:rt}

A regression tree \citep{breiman1984} recursively partitions the prediction space into nodes or boxes $B_1,\dots,B_l,\dots,B_L$ so that the units within each box are homogeneous with respect to $y$. Fitted to the sample data through a design-consistent algorithm \citep{toth2011}, it gives rise to a special case of $\Yhat^{\mathrm{PRED}}$ in which $\mh$ is the sample fit of the tree. The resulting estimator \citep{mcconville2019} can be read as a post-stratified estimator with automatically chosen cells:
\begin{equation}\label{eq:rt}
\Yhat^{\mathrm{RT}}=\sum_{l=1}^{L}\frac{N_l}{\Nhat_l}\sum_{i\in s_l} w_i\, y_i
=\sum_{i\in s} w_i^*\, y_i,\end{equation}

where $w_i^*=\frac{N_l}{\Nhat_l}\, w_i \ \ (i\in s_l)$, $s_l=s\cap B_l$, $N_l$ is the population count of node $l$, and $\Nhat_l=\sum_{i\in s_l} w_i$ its Horvitz--Thompson estimator. The weights $w_i^*$ estimate the node counts without error and satisfy $N=\sum_{i\in s} w_i^*$; in particular, they \emph{do not depend on $y$}, so the estimator is compatible with the uni-weight approach.

Compared with classical post-stratification, the tree performs automatic selection: instead of building all possible interactions among qualitative variables, it includes only the variables (and interactions) that explain the variability of $y$, reducing the dimension of the model. The estimator can also be viewed as a calibrated estimator whose calibration equation is determined automatically by the algorithm. A variance estimator is obtained from \eqref{eq:varpred} by replacing the residuals by
\begin{equation}\label{eq:rt-resid}
e_i=y_i-\Bigl(\sum_{i\in s_l} w_i\Bigr)^{-1}\Yhat^{\mathrm{HT}}_l ,
\end{equation}
that is, the residual with respect to the weighted mean of the node to which the unit belongs. The estimator is asymptotically unbiased; the size of the tree (number of boxes) governs its variance, as illustrated in Section~\ref{sec:sim}.

\begin{proposition}[Design consistency of the tree-assisted estimator]\label{prop:rt}
Consider a sequence of nested finite populations and sampling designs in the standard finite-population asymptotic framework \citep{sarndal1992}, and assume:
\begin{enumerate}[label=\textup{(A\arabic*)},leftmargin=2.8em]
\item \textup{(moments)} the study variable and the auxiliary vector have uniformly bounded fourth moments across the sequence;
\item \textup{(design)} the inclusion probabilities satisfy $\lambda\le N\pi_i\le\Lambda$ for constants $0<\lambda\le\Lambda<\infty$, and the design admits a design-consistent Horvitz--Thompson variance estimator;
\item \textup{(tree)} the recursive-partitioning algorithm yields a sample-fitted partition that is design consistent for a fixed limiting partition $\{B_l\}_{l=1}^{L}$ in the sense of \citet{toth2011}, with $L$ bounded and expected within-node sample sizes $\E(n_l)\to\infty$.
\end{enumerate}
Then $\Yhat^{\mathrm{RT}}$ is design consistent, $N^{-1}(\Yhat^{\mathrm{RT}}-Y)=O_p(n^{-1/2})$, and the variance estimator defined by \eqref{eq:varpred}--\eqref{eq:rt-resid} is design consistent for the asymptotic variance of $\Yhat^{\mathrm{RT}}$.
\end{proposition}

\begin{proof}[Proof sketch]
Conditionally on the partition, $\Yhat^{\mathrm{RT}}$ in \eqref{eq:rt} is the post-stratified estimator on the cells $\{B_l\}$. By (A3) the sample-fitted partition converges to the limiting partition, so $\Yhat^{\mathrm{RT}}$ differs from the post-stratified estimator on $\{B_l\}$ by $o_p(n^{-1/2})$ \citep{toth2011}. On the limiting partition the assisting function $m(\bx)=\sum_{l}\mathbf{1}\{\bx\in B_l\}\,\bar y_{B_l}$ (the population node means) is fixed and non-random, so $\Yhat^{\mathrm{RT}}$ is asymptotically equivalent to the difference estimator \eqref{eq:diff} with this $m$. The difference estimator is design unbiased with the variance given in Section~\ref{sec:asistida}, and the plug-in residuals \eqref{eq:rt-resid} yield a design-consistent variance estimator under (A1)--(A2). This coincides with the result of \citet{mcconville2019}, who establish design consistency and asymptotic normality of the regression-tree estimator building on the design-consistency of sample-fitted trees of \citet{toth2011}.
\end{proof}

\begin{remark}
Condition (A3) makes explicit the trade-off illustrated in Table~\ref{tab:U}: keeping $L$ bounded (or letting it grow slowly with $n$) preserves $\E(n_l)\to\infty$ and controls the Kish effect, whereas over-grown trees violate it and inflate the variance.
\end{remark}

\section{Domain estimation}\label{sec:dominios}

\subsection{Ground rules}

Let $U_1,\dots,U_d,\dots,U_D$ be disjoint domains with $U=\bigcup_{d=1}^{D}U_d$ and $N=\sum_{d=1}^{D}N_d$. We assume \emph{unplanned} domains (they are not design strata), so the domain sample size $n_d=|s\cap U_d|$ is random. The goal is to estimate
\begin{equation}
Y_d=\sum_{i\in U_d} y_i ,
\end{equation}
by means of \emph{direct} estimators $\Yhat_d=\sum_{i\in s_d} w_i y_i$, with $s_d=s\cap U_d$.

The extended variable $y_{d}$ is useful, defined by $y_{di}=y_i$ if $i\in U_d$ and $y_{di}=0$ otherwise, so that $Y_d=\sum_{i\in U} y_{di}$ and the Horvitz--Thompson estimator for the domain is
\begin{equation}\label{eq:htd}
\Yhat_d^{\mathrm{HT}}=\sum_{i\in s_d} w_i y_i=\sum_{i\in s} w_i y_{di},
\end{equation}
unbiased, with a variance estimator obtained by replacing $y$ with $y_d$ in \eqref{eq:ht}.

\subsection{Two prediction strategies for domains}

\paragraph{Strategy 1: population-level model (RT$_U$).} A single model is fitted with the whole sample $s$, and the domain estimator is built with the predictions $\mh(\bx_i)$, $i\in U_d$, plus an adjustment term that protects the estimator:
\begin{equation}\label{eq:predU}
\Yhat_d^{\mathrm{PRED},U}=\sum_{i\in U_d}\mh(\bx_i)+\sum_{i\in s_d} w_i\bigl(y_i-\mh(\bx_i)\bigr).
\end{equation}
This estimator is \emph{not direct} in the strict sense (it uses $\mh$ fitted outside the domain), but it satisfies the additivity or \emph{benchmarking} property with respect to the population total, creates economies of scale (a single algorithm), and when $\mh$ comes from a GREG or a tree, gives rise to a single set of weights $w_i^*$ applicable to any domain (uni-weight). Its weakness is that the variance reduction relative to \eqref{eq:htd} may be null if the domain has its own characteristics not captured by the global model.

\paragraph{Strategy 2: domain-level model (RT$_D$).} A specific model $\mh_d(\cdot)$ is fitted using only the data in $s_d$:
\begin{equation}\label{eq:predD}
\Yhat_d^{\mathrm{PRED},D}=\sum_{i\in U_d}\mh_d(\bx_i)+\sum_{i\in s_d} w_i\bigl(y_i-\mh_d(\bx_i)\bigr).
\end{equation}
This estimator \emph{is direct}, can be expressed as $\sum_{i\in s_d} w_i^* y_i$ when the method allows it (GREG, trees), and requires that $\E(n_d)$ be large enough to estimate $\mh_d(\cdot)$ stably. When the domain has its own characteristics, the variance reduction relative to \eqref{eq:predU} can be substantial.

\subsection{Domain-level regression tree}

A domain-specific regression tree $\mh_d(\cdot)$ is defined, which partitions a reduced prediction space into nodes $B_{d1},\dots,B_{dl},\dots,B_{dL_d}$. The estimator of the domain total is
\begin{equation}\label{eq:rtd}
\Yhat_d^{\mathrm{RT}}=\sum_{i\in U_d}\mh_d(\bx_i)+\sum_{i\in s_d} w_i\bigl(y_i-\mh_d(\bx_i)\bigr)
=\sum_{i\in s_d} w_i^*\, y_i=\sum_{i\in s_d}\frac{N_{dl}}{\Nhat_{dl}}\, w_i\, y_i ,
\end{equation}
with $w_i^*=(N_{dl}/\Nhat_{dl})w_i$ that estimate the domain node counts without error and satisfy
\begin{equation}
\sum_{i\in s_d} w_i^*=N_d .
\end{equation}
If the domains are disjoint, the collection $\{w_i^*\}_{i\in s}$ obtained by pooling the weights of all domains constitutes a single weighting system: it can be used to estimate at the level of $U$, $\Yhat=\sum_{d=1}^{D}\sum_{i\in s_d} w_i^* y_i$, for different parameters, and for different domains. The price is a \emph{trade-off}: obtaining more precise estimates within each domain does not imply having the ``optimal'' estimator at the population level (for example, relative to $\Yhat^{\mathrm{RT}}$ with a single global tree). In large-scale surveys this cost is usually minor.

\subsection{Why the population-level model behaves like Horvitz--Thompson within domains}\label{sec:colapso}

The contrasting behavior of the two strategies admits a transparent explanation in the linear case, which extends heuristically to the tree. If the model is linear and estimated at the level of $U$, the domain estimator can be written as
\begin{equation}\label{eq:colapso}
\Yhat_d^{*}=\sum_{i\in s_d} w_i^* y_i=\Yhat_d^{\mathrm{HT}}+(\bm{X}-\Xhat)\T\widehat{\bm{R}}_U,
\end{equation}
where 
$$\bm{R}_U=\Bigl(\sum_{i\in s} w_i\bx_i\bx_i\T\Bigr)^{-1}\Bigl(\sum_{i\in s} w_i\bx_i\, y_{di}\Bigr),$$
with 
$$w_i^*=w_i\bigl[1+(\bm{X}-\Xhat^{\mathrm{HT}})\T(\sum_{i\in s} w_i\bx_i\bx_i\T)^{-1}\bx_i\bigr].$$
A variance estimator for $\Yhat_d^{*}$ is
\begin{equation}
\Vhat(\Yhat_d^{*})=\sum_{i\in s}\sum_{j\in s}\pi_{ij}^{-1}(\pi_{ij}-\pi_i\pi_j)\,\frac{e_i}{\pi_i}\,\frac{e_j}{\pi_j},
\end{equation}
where, crucially, the errors are
\begin{equation}\label{eq:errs}
e_i=
\begin{cases}
y_i-\bx_i\T\bm{R}_U, \ i\in U_d,\\[2pt]
-\,\bx_i\T\bm{R}_U, \ i\notin U_d .
\end{cases}
\end{equation}
The units \emph{outside} the domain contribute residuals $-\bx_i\T\bm{R}_U$ that do not vanish: the population fit was not designed to make the variability within $U_d$ small. As a consequence, the residual sum of squares is not reduced relative to the Horvitz--Thompson estimator, and $\Yhat_d^*$ performs similarly to $\Yhat_d^{\mathrm{HT}}$, which uses no auxiliary information within the domains of interest. As \citet{estevao1999} already noted, the level (population, domain, or intermediate) at which the auxiliary information is incorporated is a critical factor for the variance of the uni-weight estimator.

The argument carries over to the tree: a tree fitted at the level of $U$ defines nodes that are homogeneous with respect to $y$ \emph{in the population}, but those nodes need not be homogeneous \emph{within} a domain that has its own structure. By contrast, the domain-specific tree seeks nodes that are homogeneous in $U_d$, which is the source of RT$_D$'s precision gains.

\begin{proposition}[The level at which the model is fitted]\label{prop:level}
Under the conditions of Proposition~\ref{prop:rt} applied to the extended variable $y_d$, and writing $\mathrm{AV}(\cdot)$ for the asymptotic design variance:
\begin{enumerate}[label=\textup{(\alph*)},leftmargin=2.4em]
\item \textup{(population-level fit, RT$_U$)} if the assisting model is fitted on all of $s$, then, because $y_{di}=0$ for every $i\notin U_d$, the fitted relationship for $y_d$ is attenuated and
\[
\mathrm{AV}\!\bigl(\Yhat_d^{\mathrm{PRED},U}\bigr)=\mathrm{AV}\!\bigl(\Yhat_d^{\mathrm{HT}}\bigr)\,(1+o(1)),
\]
i.e.\ there is asymptotically no efficiency gain over Horvitz--Thompson for the domain total;
\item \textup{(domain-level fit, RT$_D$)} if the assisting model is fitted within the domain using $s_d$, with $\E(n_d)\to\infty$ and within-domain predictive power $R_d^2\in(0,1]$, then
\[
\mathrm{AV}\!\bigl(\Yhat_d^{\mathrm{PRED},D}\bigr)=\bigl(1-R_d^2\bigr)\,\mathrm{AV}\!\bigl(\Yhat_d^{\mathrm{HT}}\bigr)\,(1+o(1))\le \mathrm{AV}\!\bigl(\Yhat_d^{\mathrm{HT}}\bigr)(1+o(1)),
\]
with the reduction governed by $R_d^2$.
\end{enumerate}
\end{proposition}

\begin{proof}[Proof sketch]
Both estimators target $Y_d=\sum_{i\in U} y_{di}$, and their asymptotic variances are governed by the population residuals $e_i$ of the assisting model for $y_d$, as in \eqref{eq:errs}.
(a) When the model is fitted on all of $U$, the regression response $y_d$ is identically zero on $U\setminus U_d$, a fraction $1-N_d/N$ of the population; the fit (least squares or tree) is pulled toward the null model, the coefficient $\bm R_U$ is attenuated, and the residual variation $\sum_{i\in U} e_i^2/\pi_i$ equals $\sum_{i\in U} y_{di}^2/\pi_i$ up to a term of smaller order. Since the latter is the leading term of $\mathrm{AV}(\Yhat_d^{\mathrm{HT}})$, the variance ratio tends to $1$.
(b) When the model is fitted within $U_d$, the assisting function captures the within-domain relationship between $y$ and $\bx$; the domain residuals satisfy $\sum_{i\in U_d}(y_i-\mh_d(\bx_i))^2/\pi_i=(1-R_d^2)\sum_{i\in U_d} y_i^2/\pi_i+o(\cdot)$, while by construction the extended residuals vanish on $U\setminus U_d$. Hence the leading variance term shrinks by the factor $1-R_d^2$.
\end{proof}

\begin{remark}
For regression trees the same dichotomy holds with $R_d^2$ replaced by the proportion of within-domain variance of $y$ explained by the partition. A globally fitted tree assigns to each unit the \emph{population} node mean; restricted to $U_d$ these means are generally biased for the domain node means, so the within-domain residual variance is not reduced the mechanism behind the $\texttt{deff}\approx 1$ of RT$_U$ in Section~\ref{sec:sim}. A domain-specific tree instead splits to make the nodes homogeneous within $U_d$.
\end{remark}

\section{Simulation study}\label{sec:sim}

\subsection{Experimental design}

We use public microdata from the Uruguayan Continuous Household Survey (ECH). The population $U$ consists of working-age individuals. Two dichotomous variables of interest are considered: $y^{(1)}=1$ if the person is employed (and $0$ otherwise) and $y^{(2)}=1$ if the person is unemployed (and $0$ otherwise). The target parameters are the totals of employed and unemployed persons, $Y^{(1)}=\sum_{i\in U} y^{(1)}_i$ and $Y^{(2)}=\sum_{i\in U} y^{(2)}_i$.

The auxiliary variables $\bx$, available at the micro level and a priori good predictors, are: single-year age, sex (2 categories), highest educational level attained (5 categories), and race (3 categories).

\subsection{Population-level results}

$R=1000$ samples (replicates) are drawn under simple random sampling of size $n=1000$. For each replicate, $\Yhat^{\mathrm{HT}}$ and the tree-assisted estimators are computed, repeating the experiment for different tree sizes (average number of boxes). We report the relative bias $\mathrm{relbias}(\Yhat)=R^{-1}\sum_{r=1}^{R}\Yhat^{(r)}/Y$, so that a value of $1$ indicates design unbiasedness; the design effect relative to Horvitz--Thompson, $\mathrm{deff}(\Yhat)=\Var_{\mathrm{MC}}(\Yhat)/\Var_{\mathrm{MC}}(\Yhat^{\mathrm{HT}})$, where $\Var_{\mathrm{MC}}$ is the Monte Carlo variance across the $R$ replicates; and the Kish design effect due to weighting, $\mathrm{deffK}=1+\mathrm{CV}^2(w^*)$, with $\mathrm{CV}(w^*)$ the coefficient of variation of the final weights.

\begin{table}[ht]
\centering
\caption{Performance of the tree-assisted estimator at the population level, by tree size. Modeled variable: employed $y^{(1)}$; non-modeled variable: unemployed $y^{(2)}$.}
\label{tab:U}
\begin{tabular}{rccccc}
\toprule
\# boxes (avg) & \texttt{relbias}($\Yhat^{(1)}$) & \texttt{relbias}($\Yhat^{(2)}$) & \texttt{deff}($\Yhat^{(1)}$) & \texttt{deff}($\Yhat^{(2)}$) & \texttt{deffK} \\
\midrule
140 & 0.998 & 1.014 & 0.731 & 1.266 & 1.147 \\
90  & 0.999 & 1.017 & 0.630 & 1.099 & 1.088 \\
60  & 0.998 & 0.997 & 0.649 & 1.061 & 1.062 \\
50  & 1.000 & 1.007 & 0.635 & 1.099 & 1.049 \\
40  & 0.999 & 1.007 & 0.638 & 1.026 & 1.040 \\
35  & 0.999 & 1.007 & 0.647 & 1.019 & 1.033 \\
30  & 0.999 & 1.002 & 0.584 & 0.964 & 1.028 \\
20  & 1.000 & 0.998 & 0.546 & 0.992 & 1.021 \\
10  & 1.001 & 0.997 & 0.667 & 1.023 & 1.012 \\
\bottomrule
\end{tabular}
\end{table}

\noindent Table~\ref{tab:U} shows that $\Yhat^{\mathrm{RT}}_{(1)}$ is approximately unbiased for any tree size. Performance in terms of variance for the modeled variable is clearly superior to that of Horvitz--Thompson (\texttt{deff} around $0.55$--$0.73$), but it begins to deteriorate as the tree grows, owing to the increase in the Kish effect (\texttt{deffK}). For the \emph{non-modeled} variable ($y^{(2)}$), although the covariates were good predictors of both variables, the improvement relative to Horvitz--Thompson is practically null (\texttt{deff}$\approx 1$). This is an intrinsic limit of the uni-weight approach when the model is fitted for a particular variable.

\paragraph{Comparison with post-stratification and raking.} $\Yhat^{\mathrm{RT}}$ is compared with a post-stratified estimator (PS) whose $G=96$ cells arise from the cross of age band (4, arbitrary), sex (2), educational level (4), and race (3) omitting cells with $n_g=0$, which introduces bias and with raking (incomplete post-stratification).

\begin{table}[ht]
\centering
\caption{Comparison of population-level estimators for the employed variable $y^{(1)}$.}
\label{tab:rtvsps}
\begin{tabular}{lcc}
\toprule
Estimator & \texttt{relbias}($\Yhat^{(1)}$) & \texttt{deff}($\Yhat^{(1)}$) \\
\midrule
PS      & 0.972 & 0.610 \\
RAKING  & 0.999 & 0.579 \\
RT      & 0.999 & 0.548 \\
\bottomrule
\end{tabular}
\end{table}

\noindent The tree-assisted estimator achieves the lowest \texttt{deff} and, unlike PS, incurs neither the bias associated with empty cells nor the loss of coherence/comparability of post-stratification with many cells. The empirical distribution of the three estimators shows that PS is shifted relative to RT and raking, consistent with its larger bias.

\begin{figure}[ht]
\centering
\includegraphics[width=0.8\textwidth]{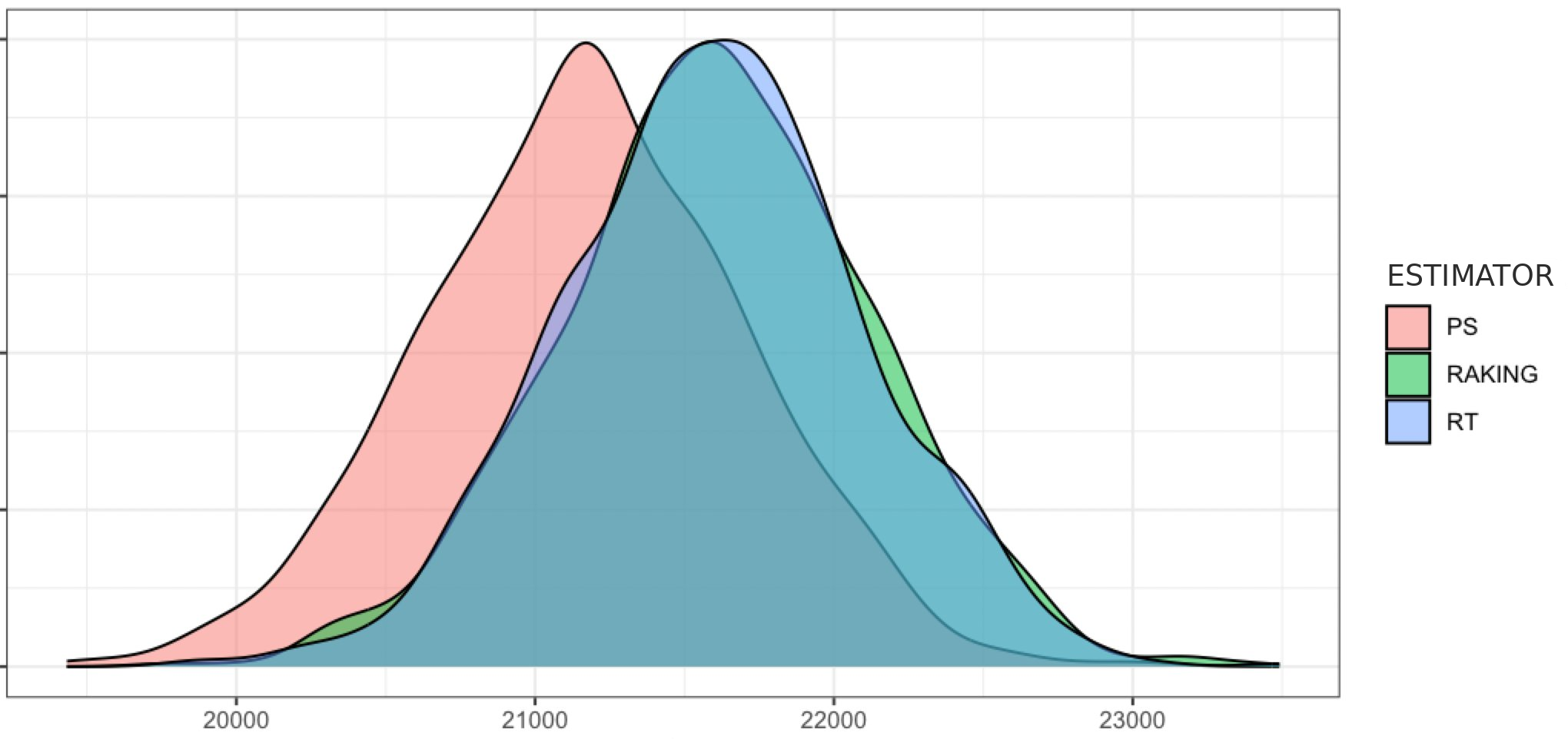}
\caption{Empirical distribution of the estimators (PS, RAKING, RT)}
\label{fig:conceptual}
\end{figure}

\subsection{Domain results}

A similar scenario is replicated with minor changes: the goal is to estimate the total of employed persons by department, $Y_d=\sum_{i\in U_d} y_i$, and the country total $Y=\sum_{d=1}^{D} Y_d$. Samples of size $n=10000$ are drawn under simple random sampling; the domains are unplanned, with $\E(n_d)=(N_d/N)\,n$ large enough that the direct estimators do not begin to fail. For each replicate, three estimators are computed:
\begin{itemize}[leftmargin=2.2em]
\item \textbf{HT}: $\Yhat_d^{\mathrm{HT}}=(N/n)\sum_{i\in s_d} y_i$, and $\Yhat^{\mathrm{HT}}=(N/n)\sum_{i\in s} y_i$.
\item \textbf{RT$_U$}: tree estimated at the $U$ level, with uni-weight weights; $\Yhat_d^{\mathrm{RT}*}=\sum_{i\in s_d} w_i^* y_i$.
\item \textbf{RT$_D$}: domain-specific tree; $\Yhat_d^{\mathrm{RT},D}=\sum_{i\in s_d} w_i^* y_i$, and $\Yhat=\sum_{d=1}^{D}\sum_{i\in s_d} w_i^* y_i$.
\end{itemize}

\noindent The results are as expected from Section~\ref{sec:colapso}: the RT$_U$ estimator exhibits \texttt{deff} close to $1$ across all domains, that is, it does not improve on Horvitz--Thompson whereas RT$_D$ achieves substantial reductions, with \texttt{deff} around $0.3$ in most departments. The advantage of RT$_D$ holds across the range of expected domain sample sizes.

\begin{figure}[!htbp]
\centering
\includegraphics[width=0.8\textwidth]{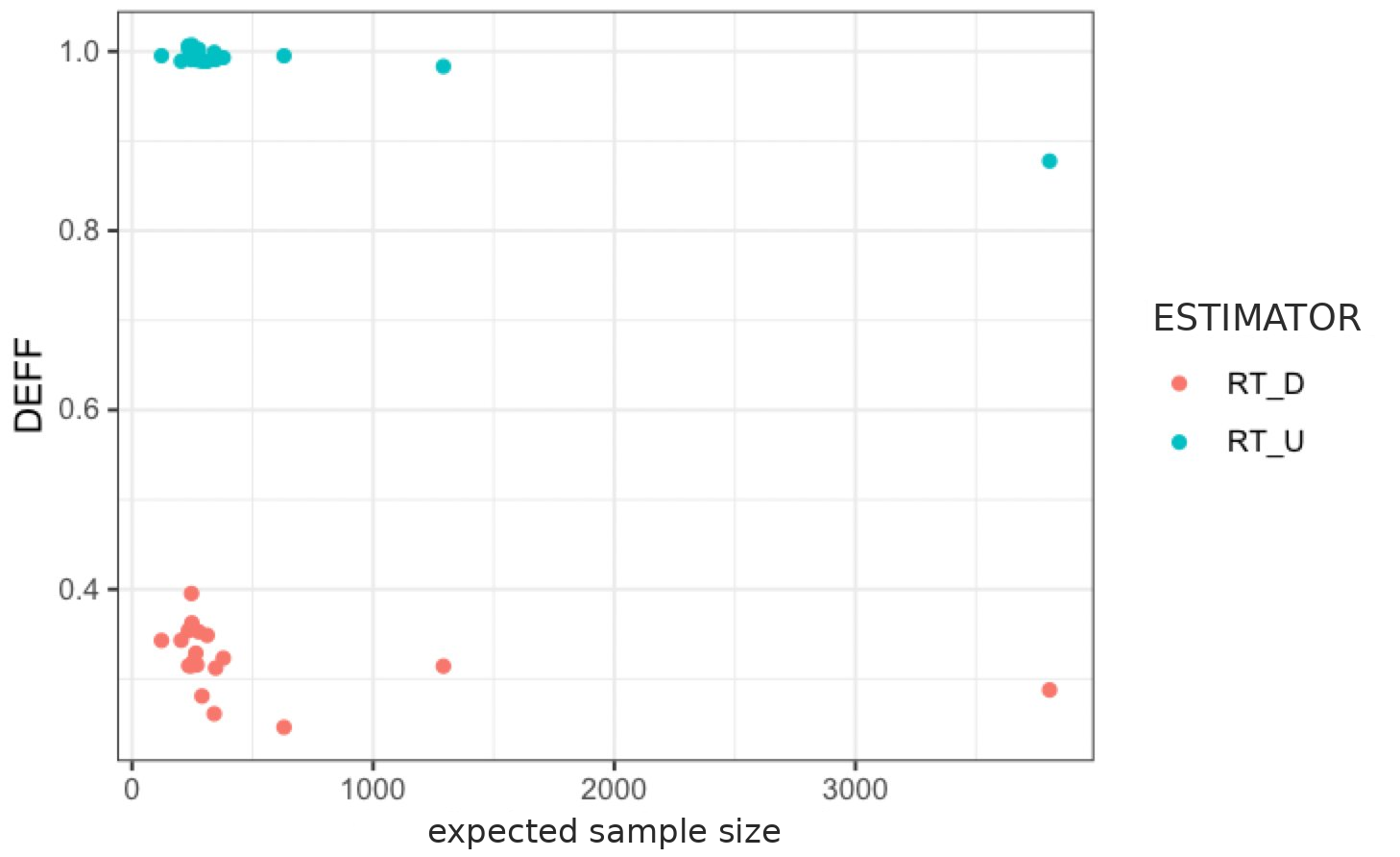}
\caption{Design effect (deff) relative to Horvitz–Thompson versus expected domain sample size, by department, for RT$_U$ and RT$_D$}
\label{fig:conceptual}
\end{figure}

\FloatBarrier
\section{Discussion and conclusions}\label{sec:disc}

The estimator assisted by a domain-specific regression tree (RT$_D$) offers, within the design-based direct framework and the uni-weight approach, an attractive route to improving the precision of estimates in unplanned domains. It is asymptotically unbiased, admits a representation as a weighted sum with weights independent of $y$ that estimate the domain count without error, and, by construction at the level of each disjoint domain, gives rise to a single weighting system applicable to the population and to different parameters. The experiment with the ECH confirms that, whereas the global tree (RT$_U$) brings no gains relative to Horvitz--Thompson within domains, the domain-specific tree (RT$_D$) markedly reduces the variance, in line with the residual mechanism of Section~\ref{sec:colapso} and with the classical result of \citet{estevao1999,estevao2004}.

The approach has limits worth making explicit. First, the improvement is concentrated in the variable used to build the model: for a \emph{non-modeled} variable of interest, the uni-weight estimator based on a tree fitted for another variable does not necessarily improve on Horvitz--Thompson, as Table~\ref{tab:U} shows. Second, growing the tree increases the Kish effect and may deteriorate the variance. Third, the estimator requires $\E(n_d)$ to be large enough; in very small domains, direct estimators lose stability and indirect model-based estimation becomes competitive again \citep{sarndal1984,rao2015}. Finally, the local gain is obtained at the cost of not having the optimal estimator at the population level, although in large-scale surveys this cost is usually minor.

Several extensions follow naturally. The domain-specific construction carries over to random forests \citep{dagdoug2023}, which may capture within-domain structure more flexibly at the cost of greater weight variability. Between the two extremes studied here lies a family of intermediate strategies in which the tree is fitted at a level of aggregation above the domain but below the population, for instance, at the level of geographic regions grouping several departments in line with the ``intermediate level'' already identified by \citet{estevao1999} as a determinant of the variance of the uni-weight estimator. Finally, the choice among RT$_U$, RT$_D$, and indirect model-based estimation can be cast as a decision rule indexed by the expected domain sample size $\E(n_d)$: RT$_D$ where $\E(n_d)$ is large enough for stable within-domain fitting, indirect SAE in the small-$\E(n_d)$ regime \citep{rao2015}, and RT$_U$ (or a global model) where benchmarking to the population total is the binding requirement.

\bibliographystyle{plainnat}
\bibliography{referencias}

\end{document}